\documentclass[%
 preprint,
 pra,
 amsmath,amssymb,
 aps,
]{revtex4-1}

\usepackage{graphicx}
\usepackage{dcolumn}
\usepackage{bm}

\usepackage{amsmath,amsfonts,amssymb,dsfont,caption,hyperref,color,epsfig,graphics,graphicx,latexsym,mathrsfs,revsymb,theorem,url,verbatim,epstopdf,cleveref,dashrule,cases,booktabs,diagbox,threeparttable,setspace,natbib,mathdots}

\hypersetup{colorlinks,linkcolor={blue},citecolor={blue},urlcolor={red}}

\newtheorem{definition}{Definition}

\newtheorem{lemma}[definition]{Lemma}

\newtheorem{theorem}[definition]{Theorem}
\newtheorem{corollary}[definition]{Corollary}

\def\squareforqed{$\square$}
\def\qed{\ifmmode\squareforqed\else{\unskip\nobreak\hfil
\penalty50\hskip1em\null\nobreak\hfil\squareforqed
\parfillskip=0pt\finalhyphendemerits=0\endgraf}\fi}

\def\endenv{\ifmmode\;\else{\unskip\nobreak\hfil
\penalty50\hskip1em\null\nobreak\hfil\;
\parfillskip=0pt\finalhyphendemerits=0\endgraf}\fi}

\newenvironment{proof}{\noindent \textbf{{Proof.~} }}{\qed}

\def\bpf{\begin{proof}}
\def\epf{\end{proof}}
\def\bea{\begin{eqnarray}}
\def\eea{\end{eqnarray}}
\def\beq{\begin{equation}}
\def\eeq{\end{equation}}
\def\bal{\begin{aligned}}
\def\eal{\end{aligned}}
\def\bma{\begin{bmatrix}}
\def\ema{\end{bmatrix}}

\def\tr{{\rm Tr}}

\def\a{\alpha}
\def\b{\beta}
\def\g{\gamma}

\def\r{\rho}

\newcommand{\trace}[1]{\mathrm{Tr}(#1)}

\newcommand{\nc}{\newcommand}

\nc{\bbA}{\mathbb{A}} \nc{\bbB}{\mathbb{B}} \nc{\bbC}{\mathbb{C}}
\nc{\bbD}{\mathbb{D}} \nc{\bbE}{\mathbb{E}} \nc{\bbF}{\mathbb{F}}
\nc{\bbG}{\mathbb{G}} \nc{\bbH}{\mathbb{H}} \nc{\bbI}{\mathbb{I}}
\nc{\bbJ}{\mathbb{J}} \nc{\bbK}{\mathbb{K}} \nc{\bbL}{\mathbb{L}}
\nc{\bbM}{\mathbb{M}} \nc{\bbN}{\mathbb{N}} \nc{\bbO}{\mathbb{O}}
\nc{\bbP}{\mathbb{P}} \nc{\bbQ}{\mathbb{Q}} \nc{\bbR}{\mathbb{R}}
\nc{\bbS}{\mathbb{S}} \nc{\bbT}{\mathbb{T}} \nc{\bbU}{\mathbb{U}}
\nc{\bbV}{\mathbb{V}} \nc{\bbW}{\mathbb{W}} \nc{\bbX}{\mathbb{X}}
\nc{\bbY}{\mathbb{Y}} \nc{\bbZ}{\mathbb{Z}}

\nc{\bA}{{\bf A}} \nc{\bB}{{\bf B}} \nc{\bC}{{\bf C}}
\nc{\bD}{{\bf D}} \nc{\bE}{{\bf E}} \nc{\bF}{{\bf F}}
\nc{\bG}{{\bf G}} \nc{\bH}{{\bf H}} \nc{\bI}{{\bf I}}
\nc{\bJ}{{\bf J}} \nc{\bK}{{\bf K}} \nc{\bL}{{\bf L}}
\nc{\bM}{{\bf M}} \nc{\bN}{{\bf N}} \nc{\bO}{{\bf O}}
\nc{\bP}{{\bf P}} \nc{\bQ}{{\bf Q}} \nc{\bR}{{\bf R}}
\nc{\bS}{{\bf S}} \nc{\bT}{{\bf T}} \nc{\bU}{{\bf U}}
\nc{\bV}{{\bf V}} \nc{\bW}{{\bf W}} \nc{\bX}{{\bf X}}
\nc{\bY}{{\bf Y}} \nc{\bZ}{{\bf Z}}

\nc{\bmA}{{\bm A}} \nc{\bmB}{{\bm B}} \nc{\bmC}{{\bm C}}
\nc{\bmD}{{\bm D}} \nc{\bmE}{{\bm E}} \nc{\bmF}{{\bm F}}
\nc{\bmG}{{\bm G}} \nc{\bmH}{{\bm H}} \nc{\bmI}{{\bm I}}
\nc{\bmJ}{{\bm J}} \nc{\bmK}{{\bm K}} \nc{\bmL}{{\bm L}}
\nc{\bmM}{{\bm M}} \nc{\bmN}{{\bm N}} \nc{\bmO}{{\bm O}}
\nc{\bmP}{{\bm P}} \nc{\bmQ}{{\bm Q}} \nc{\bmR}{{\bm R}}
\nc{\bmS}{{\bm S}} \nc{\bmT}{{\bm T}} \nc{\bmU}{{\bm U}}
\nc{\bmV}{{\bm V}} \nc{\bmW}{{\bm W}} \nc{\bmX}{{\bm X}}
\nc{\bmY}{{\bm Y}} \nc{\bmZ}{{\bm Z}}

\nc{\cA}{{\cal A}} \nc{\cB}{{\cal B}} \nc{\cC}{{\cal C}}
\nc{\cD}{{\cal D}} \nc{\cE}{{\cal E}} \nc{\cF}{{\cal F}}
\nc{\cG}{{\cal G}} \nc{\cH}{{\cal H}} \nc{\cI}{{\cal I}}
\nc{\cJ}{{\cal J}} \nc{\cK}{{\cal K}} \nc{\cL}{{\cal L}}
\nc{\cM}{{\cal M}} \nc{\cN}{{\cal N}} \nc{\cO}{{\cal O}}
\nc{\cP}{{\cal P}} \nc{\cQ}{{\cal Q}} \nc{\cR}{{\cal R}}
\nc{\cS}{{\cal S}} \nc{\cT}{{\cal T}} \nc{\cU}{{\cal U}}
\nc{\cV}{{\cal V}} \nc{\cW}{{\cal W}} \nc{\cX}{{\cal X}}
\nc{\cY}{{\cal Y}} \nc{\cZ}{{\cal Z}}

\nc{\hA}{{\hat{A}}} \nc{\hB}{{\hat{B}}} \nc{\hC}{{\hat{C}}}
\nc{\hD}{{\hat{D}}} \nc{\hE}{{\hat{E}}} \nc{\hF}{{\hat{F}}}
\nc{\hG}{{\hat{G}}} \nc{\hH}{{\hat{H}}} \nc{\hI}{{\hat{I}}}
\nc{\hJ}{{\hat{J}}} \nc{\hK}{{\hat{K}}} \nc{\hL}{{\hat{L}}}
\nc{\hM}{{\hat{M}}} \nc{\hN}{{\hat{N}}} \nc{\hO}{{\hat{O}}}
\nc{\hP}{{\hat{P}}} \nc{\hR}{{\hat{R}}} \nc{\hS}{{\hat{S}}}
\nc{\hT}{{\hat{T}}} \nc{\hU}{{\hat{U}}} \nc{\hV}{{\hat{V}}}
\nc{\hW}{{\hat{W}}} \nc{\hX}{{\hat{X}}} \nc{\hZ}{{\hat{Z}}}

\nc{\hn}{{\hat{n}}}

\def\tr{\mathop{\rm Tr}}

\begin{document}

\title{Superadditivity of Convex Roof Coherence Measures in Multipartite System}

\author{Honglin Ren}\email[]{2409136@buaa.edu.cn}
\affiliation{LMIB(Beihang University), Ministry of education, and School of Mathematical Sciences, Beihang University, Beijing 100191, China} 

\author{Lin Chen}\email[]{linchen@buaa.edu.cn (corresponding author)}
\affiliation{LMIB(Beihang University), Ministry of education, and School of Mathematical Sciences, Beihang University, Beijing 100191, China}
\begin{abstract}
In this paper, we investigate the convex roof measure of quantum coherence, with a focus on their superadditive properties. We propose sufficient conditions and establish a framework for coherence superadditivity in tripartite and multipartite systems. Through theoretical derivation, the relevant theorems are given. These results not only expand our understanding of the superadditivity of pure and mixed states but also characterize the conditions under which the superadditivity relations reach equality. Finally, the proposed methods and conclusions are verified through representative examples, providing new theoretical insights into the distribution of quantum coherence in multi-part systems.
\end{abstract}

\maketitle


\section{Introduction}
\label{sec:int}

Coherence is a fundamental feature of quantum theory and plays a central role in quantum information processing and foundational physics. The foundational achievements of quantum information science have revealed the powerful potential of quantum resources such as entanglement and coherence. For example, quantum cryptography based on entanglement correlations \cite{ekert1991quantum}, quantum teleportation \cite{bennett1993teleporting}, and Shor’s polynomial-time quantum algorithm \cite{shor1997polynomial}. All have demonstrated the non-classical advantages of quantum systems in information processing. These achievements have been systematically expounded in classic textbooks and reviews of quantum information theory \cite{nielsen2000quantum,plenio2014introduction,nielsen2010quantum}. Coherence was soon recognized not only as a formal property, but also as a useful physical resource. This inspired a resource-theoretic perspective, which formalized coherence measures that met rigorous axioms \cite{baumgratz2014quantifying}. Later, an operational approach to coherence as a resource was developed \cite{winter2016operational}, and coherence was also studied in broader resource theories \cite{streltsov2017colloquium}. Building on this framework, researchers explored various directions: coherence as a source of intrinsic randomness \cite{yuan2015intrinsic}, its transformation into quantum correlations \cite{ma2016converting}, and its behavior under decoherence processes \cite{zhang2018quantum}. In addition, scholars have proposed a coherence measurement method based on skew information \cite{sheng2021applications} and coherence vectors \cite{Bosyk2020}. A particularly intricate area of inquiry concerns coherence in multipartite systems. Although efforts have been made to understand the coherence distribution and superadditivity in general settings. Notable developments include the analysis of convex roof measures of coherence and their nontrivial superadditive behavior \cite{Liu2018superadditivity}, as well as advances in separability-based characterizations \cite{Cao221separability}. Recent theoretical progress has introduced new sufficient conditions for superadditivity in the coherence of formation and convex roof constructions. This includes studies on multipartite coherence distribution and generalized convex roof quantification frameworks \cite{Zhu2024}. Further work has extended analytical tools to compute such measures in large systems \cite{Shi2025} and formal unification across resource theories \cite{Yang2024}. New insights into coherence inequalities and multiqubit interactions have also emerged \cite{Qi2020}.

Despite these advances, a general understanding of when convex roof coherence measures exhibit superadditivity remains incomplete. This calls for a more unified treatment of superadditivity in both pure and mixed multipartite systems. In this work, our aim is to address this gap by developing general structural conditions under which superadditivity holds. We also analyze the corresponding equality cases and provide illustrative examples to demonstrate the applicability and tightness of our results. The focus of this paper is to investigate the superadditivity of convex roof coherence measures in multipartite quantum systems and the condition that the bipartite states satisfy additivity. We begin by establishing sufficient conditions under which the superadditivity holds for tripartite pure states. Specifically, we provide a structural criterion for the superadditivity of bipartite systems in Lemma \ref{lemma}. This lemma serves as the foundational component for subsequent generalizations. In Theorem \ref{thm:main}, we rigorously establish a superadditivity inequality for the convex roof coherence measure $C_f$, which states that the coherence of a tripartite pure state is lower bounded by a weighted sum of coherence measures associated with its reduced pure states.We then extend our analysis to the general multipartite setting in Theorem \ref{thm:An}. We establish a sufficient condition for superadditivity in arbitrary multipartite pure states. Furthermore, we extend these results to mixed states through convex roof constructions. In Theorem \ref{cor:f}, we investigate the cases where the superadditivity inequality is saturated. To illustrate these results, we analyze several concrete examples, including the coherence measure of formation, the coherence measure of the convex roof based on the entropy $\frac{1}{2}$ and based on the specific function $f$.

The remainder of this paper is organized as follows. In Sec. \ref{sec:pre}, we introduce the necessary definitions, notation, and known properties of convex roof coherence measures that are used throughout the paper. In Sec. \ref{sec:sup}, we present a sufficient condition for the superadditivity of convex roof coherence measures in tripartite systems and generalize the result to pure multipartite states and mixed states via convex roof extensions. Sec. \ref{sec:equ} focuses on the characterization of equality conditions, including both pure and mixed states, and provides examples when superadditivity inequalities are saturated. Finally, Sec. \ref{sec:con} concludes the paper with a summary of our contributions and discusses open questions for future research.

\section{Preliminaries}
\label{sec:pre}

Before presenting our findings, we first review some concepts that are relevant to the problem in this paper. Let $\mathcal{H}$ be a finite-dimensional Hilbert space with a fixed computational basis $\{|i\rangle\}_{i=0}^{d-1}$. The coherence of a state is measured with respect to a particular reference basis. In this passage, we denote the particular basis by $\{|i\rangle\}_{i=0}^{d-1}$ Then an incoherent state is denoted as $\delta =\sum _{i}p_{i}|i\rangle\langle i|$, with $\sum_{i}p_{i}=1$. The set of all incoherent states is denoted by $\mathcal{I}$. All other states are called coherent states, a general state is denoted as $\r$. An incoherent operation is defined as a completely positive trace-preserving map, $\Lambda(\r)=\sum_{n}K_{n}\r K_{n}^{\dag}$, with $\sum_{n}K_{n}^{\dagger}K_{n}=I$ and $K_{n}IK_{n}^{\dagger} \in \mathcal{I}$ for each $K_{n}$. This means that each $K_{n}$ maps an incoherent state to an incoherent state. With the above notation, we can give the definition of the coherence measure. 

\begin{definition}
\label{def:c}
A function $C$ can be considered as a coherence measure if it satisfies the following four conditions,
\begin{itemize}
    \item 
Positive definiteness: $ \ C(\rho)\geq 0$ and $C(\rho)=0$ if and only if $\rho \in \mathcal{I}$;
    \item Monotonicity under incoherent operations: $ \ C(\rho)\geq C(\Lambda(\rho))$ if $\Lambda$ is an incoherent operation ;
    \item Monotonicity under selective measurements on average: $\ C(\rho)\geq \sum_{n}p_{n}C(\rho_{n})$, $p_{n}=Tr(K_{n}\rho K_{n}^{\dagger})$, $\rho_{n}=K_{n}\rho K_{n}^{\dagger}/p_{n}$ and $\Lambda(\rho)=\sum_{n}K_{n}\rho K_{n}^{\dagger}$ is an incoherent operation ;
    \item Convexity: $\sum_{n}q_{n}C(\rho_{n})\geq C(\sum_{n}q_{n}\rho_{n})$ for any set of states$\{\rho_{n}\}$ and probability distribution $\{q_{n}\}$.
\end{itemize}
\end{definition}

Building on the general requirements above, we now introduce the specific class of convex roof coherence measures, which are central to our analysis.

\begin{definition}

    For a pure state $|\varphi\rangle =\sum_{i=0}^{d-1}c_{i}|i\rangle$, we define $C_{f}(\varphi):=f(|c_{0}|^{2},|c_{1}|^{2},...,|c_{d-1}|^{2})$. The extension to mixed states can be expressed as follows,
$$C_{f}(\rho)=\inf_{\{p_{i},\varphi_{i}\}}\sum_{i}p_{i}C_{f}(\varphi_{i}).$$
Here the infimum is taken over all possible ensemble decompositions $\rho=\sum_{i}p_{i}|\varphi_{i}\rangle\langle\varphi_{i}|$ with $p_{i}\geq0$ and $\sum_{i}p_{i}=1$.
\end{definition}

A coherence measure $C$ for a bipartite state is said to be superadditive if the following relation holds for all density matrices $\r_{AB}$ of a finite-dimensional system with respect to a fixed reference basis $\{|i\rangle_{A}\bigotimes |j\rangle_{B}\}$,
\begin{equation}
\label{eq:AB}
    C_{f}(\r_{AB})\geq C_{f}(\r_{A})+C_{f}(\r_{B}).
\end{equation}
Here, $\r_{A}=\tr_{B}\r_{AB}$ and $\r_{B}=\tr_{A}\r_{AB}$ are relative to the basis $\{|i\rangle_{A}\}$ and $\{|j\rangle_{B}\}$, respectively. 

First, we introduce a sufficient condition for superadditivity in the bipartite states\cite{Liu2018superadditivity}.
\begin{lemma}
\label{lemma}
A convex roof coherence measure $C_{f}$ is superadditive if the following inequality is satisfied for all pure states $|\varphi\rangle_{AB}=\sum_{ij}c_{ij}|i\rangle_{A}|j\rangle_{B}$ with $\sum_{ij}|c_{ij}|^{2}=1$,
\begin{equation}
\label{eq:2}
    C_{f}(|\varphi\rangle_{AB})\geq C_{f}(\sum_{i}\sqrt{q_{i}}|i\rangle_{A})+\sum_{i}q_{i}C_{f}(|\varphi_{i}\rangle_{B}).
\end{equation}
Here $q_{i}=\sum_{j}|c_{ij}|^{2}$ and $|\varphi_{i}\rangle_{B}=\frac{1}{\sqrt{q_{i}}}\sum_{j}c_{ij}|j\rangle_{B}$.
\end{lemma}
An alternative expression of the sufficient condition for the convex roof coherence measure $C_{f}$ to be superadditive can be stated as
\begin{equation}
    C_{f}(|\varphi\rangle_{AB})\geq \sum_{j}p_{j}C_{f}(|\varphi_{j}\rangle_{A})+\sum_{i}q_{i}C_{f}(|\varphi_{i}\rangle_{B}).
\end{equation}
Here, $p_{j}=\sum_{i}|c_{ij}|^{2}$ and $|\varphi_{j}\rangle_{A}=\frac{1}{\sqrt{p_{j}}}\sum_{i}c_{ij}|i\rangle_{A}$.

\section{Superadditivity of convex roof coherence measures}
\label{sec:sup}
Building on the definitions of coherence measures and convex roof introduced in Section II, we now investigate their superadditivity properties. Although previous studies primarily addressed bipartite cases, our analysis extends to tripartite and general multipartite states. We begin with tripartite systems and subsequently generalize the results to n-partite states. 
\subsection{superadditivity of tripartite states}
We next consider the superadditivity properties of convex roof coherence measurement in the tripartite states. 

\begin{definition}
    A coherence measure $C$ for tripartite states is said to be superadditive if the following relation is valid for all density matrices $\r_{ABC}$ on a finite-dimensional Hilbert space with respect to a particular reference basis $\{|i\rangle_{A}\otimes |j\rangle_{B}\otimes |k\rangle_{C}\}$,
\begin{equation}
\label{eq:ABC}
    C_{f}(\r_{ABC})\geq C_{f}(\r_{A})+C_{f}(\r_{B})+C_{f}(\r_{C}).
\end{equation}
Here $\r_{A}=\tr_{BC}\r_{ABC},\ \r_{B}=\tr_{AC}\r_{ABC},\ \r_{C}=\tr_{AB}\r_{ABC}$ with respect to $\{|i\rangle_{A}\otimes |j\rangle_{B}\otimes |k\rangle_{C}\}$.
\end{definition}

\begin{theorem}
\label{thm:main}
A convex roof coherence measure $C_{f}$ is superadditive if the following inequality is satisfied for all pure states $|\varphi\rangle_{ABC}=\sum_{ijk}c_{ijk}|i\rangle_{A}|j\rangle_{B}|k\rangle_{C}$ with $\sum_{ijk}|c_{ijk}|^{2}=1$,
\begin{equation}
\label{eq:a}
    C_{f}(|\varphi\rangle_{ABC})\geq \sum_{i}p_{i}C_{f}(|\a_{i}\rangle_{A})+\sum_{j}q_{j}C_{f}(|\b_{j}\rangle_{B})+\sum_{k}r_{k}C_{f}(|\g_{k}\rangle_{C}).
\end{equation}
Here $p_{i}=\sum_{jk}|c_{ijk}|^{2}$, $q_{j}=\sum_{ik}|c_{ijk}|^{2}$, $r_{k}=\sum_{ij}|c_{ijk}|^{2}$ and $|\a_{i}\rangle_{A}=\frac{1}{\sqrt{p_{i}}}\sum_{jk}c_{ijk}|i\rangle_{A}$, $|\b_{j}\rangle_{B}=\frac{1}{\sqrt{q_{j}}}\sum_{ik}c_{ijk}|j\rangle_{B}$, $|\g_{k}\rangle_{C}=\frac{1}{\sqrt{r_{k}}}\sum_{ij}c_{ijk}|k\rangle_{C}$.
\end{theorem}

\begin{proof}
We first examine Eq. \eqref{eq:a} in the case where $ \r_{ABC} $ is a pure state. We use $\r_{A}=\sum_{j}p_{j}^{'}|\psi_{j}\rangle \langle\psi_{j}|_{A}$ to represent the optimal decomposition of $\r_{A}$ that achieves the infimum in the definition of $C_{f}$. Since $\r_{A}=\sum_{i}p_{i}|\a_{i}\rangle\langle\a_{i}|_{A}$ is also an ensemble decomposition of $\r_{A}$, there must be $\sum_{i}p_{i}C_{f}(|\a_{i}\rangle_{A})\geq \sum_{i}p_{j}^{'}C_{f}(|\psi_{j}\rangle_{A})=C_{f}(\r_{A})$ according to the definition of the coherence measure. 

Similarly, we can obtain $\sum_{j}q_{j}C_{f}(|\b_{j}\rangle_{B})\geq C_{f}(\r_{B})$ and $\sum_{k}r_{k}C_{f}(|\g_{k}\rangle_{C})\geq C_{f}(\r_{C})$. This means that for pure states $\r_{ABC}$, 
\begin{equation}
\label{eq:pure}
\begin{split}
     C_{f}(\r_{ABC}) & \geq \sum_{i}p_{i}C_{f}(|\a_{i}\rangle_{A})+\sum_{j}q_{j}C_{f}(|\b_{j}\rangle_{B})+\sum_{k}r_{k}C_{f}(|\g_{k}\rangle_{C})\\
     & \geq C_{f}(\r_{A})+C_{f}(\r_{B})+C_{f}(\r_{C}).
\end{split}
\end{equation}

Second, we prove that Eq. \eqref{eq:ABC} holds for all states if it holds for pure states $\r_{ABC}$. To this end, we use $\r_{ABC}=\sum_{i}p_{i}|\psi_{i}\rangle\langle \psi_{i}|_{ABC}$ to represent one of the optimal decompositions giving $C_{f}(\r_{ABC})$. Using Eq. \eqref{eq:pure}, we have 
\begin{equation}
    C_{f}(\r_{ABC})=\sum_{i}p_{i}C_{f}(|\psi_{i}\rangle_{ABC})\geq \sum_{i}p_{i}(C_{f}((\r^{i})_{A})+C_{f}((\r^{i})_{B})+C_{f}(\r^{i})_{C})).
\end{equation}
Here $(\r^{i})_{A}=\tr_{BC}|\psi_{i}\rangle\langle\psi_{i}|_{ABC}$, $(\r^{i})_{B}=\tr_{AC}|\psi_{i}\rangle\langle\psi_{i}|_{ABC}$ and $(\r^{i})_{C}=\tr_{AB}|\psi_{i}\rangle\langle\psi_{i}|_{ABC}$.\\
According to the convexity of a coherence measure, there are $\sum_{i}p_{i}C_{f}((\r^{i})_{A})\geq C_{f}(\sum_{i}p_{i}(\r^{i})_{A})$, $\sum_{i}p_{i}C_{f}((\r^{i})_{B})\geq C_{f}(\sum_{i}p_{i}(\r^{i})_{B})$ and $\sum_{i}p_{i}C_{f}((\r^{i})_{C})\geq C_{f}(\sum_{i}p_{i}(\r^{i})_{C})$, which lead to
\begin{equation}
    C_{f}(\r_{ABC})\geq C_{f}(\sum_{i}p_{i}(\r^{i})_{A})+C_{f}(\sum_{i}p_{i}(\r^{i})_{B})+C_{f}(\sum_{i}p_{i}(\r^{i})_{C}).
\end{equation}
We note that $\r_{A}=\sum_{i}p_{i}(\r^{i})_{A}$, $\r_{B}=\sum_{i}p_{i}(\r^{i})_{B}$, and $\r_{C}=\sum_{i}p_{i}(\r^{i})_{C}$. So we  finally prove that Eq. \eqref{eq:ABC} holds for all states. This completes the proof.
\end{proof}

\subsection{superadditivity of multipartite states}

In this section, we generalize tripartite states to multipartite states about the superadditivity of the convex roof measure on the basis of Theorem  \ref{thm:main}. 

\begin{definition}
    A convex roof coherence measure $C_{f}$ of multipartite states is said to be superadditive if the following ralation is valid for all density matrics $\r_{A_{1}A_{2}...A_{n}}$ of a finite-dimensional system,
\begin{equation}
\label{eq:An}
    C_{f}(\r_{A_{1}A_{2}...A_{n}})\geq C_{f}(\r_{A_{1}})+C_{f}(\r_{A_{2}})+...+C_{f}(\r_{A_{n}})=\sum_{j=1}^{n}C_{f}(\r_{A_{j}}).
\end{equation}
Here $\r_{A_{i}}=\tr_{A_{1}...A_{i-1}A_{i+1}...A_{n}}\r_{A_{1}A_{2}...A_{n}},\ $i=1,2,...,n.
\end{definition}

\begin{theorem}
\label{thm:An}
A convex roof coherence measure $C_{f}$ is superadditive if the following inequality is satisfied for all pure states $|\a\rangle_{A_{1}A_{2}...A_{n}}=\sum_{i_{1}...i_{n}}c_{i_{1}...i_{n}}|i_{1}\rangle_{A_{1}}|i_{2}\rangle_{A_{2}}...|i_{n}\rangle_{A_{n}}$ with $\sum|c_{i_{1}...i_{n}}|^{2}=1$,
\begin{equation}
\label{eq:A}
\begin{split}
     C_{f}(|\a\rangle_{A_{1}A_{2}...A_{n}}) & \geq \sum_{i_{1}}p_{i_{1}}C_{f}(|\a_{i_{1}}\rangle_{A_{1}} )+ \sum_{i_{2}}p_{i_{2}}C_{f}(|\a_{i_{2}}\rangle_{A_{2}} )+...+\sum_{i_{n}}p_{i_{n}}C_{f}(|\a_{i_{n}}\rangle_{A_{n}} ) \\
     & =\sum_{j=1}^{n}\sum_{i_{j}}p_{i_{j}}C_{f}(|\a_{i_{j}}\rangle_{A_{j}} ) . 
\end{split}
\end{equation}
Here $p_{i_{j}}=\sum_{i_{1}...i_{j-1},i_{j+1}...,i_{n}}|c_{i_{1}...i_{n}}|^{2}$ and $|\a_{i_{j}}\rangle_{A_{j}}=\frac{1}{\sqrt{p_{i_{j}}}}\sum_{i_{1}...i_{j-1},i_{j+1}...,i_{n}}c_{i_{1}...i_{n}}|i_{j}\rangle_{A_{j}},j=1,2,...,n$.
\end{theorem}

\begin{proof}
For the pure state $\r_{A_{1}A_{2}...A_{n}}$, utilize the optimal decomposition of $\r_{A_{j}}$ that achieves the infimum in the definition of $C_{f}$. We use $\r_{A_{j}}=\sum_{i_{j}}q_{i_{j}}|\psi_{i_{j}}\rangle\langle\psi_{i_{j}}|_{A_{j}}$ to represent the optimal decomposition. Since $\r_{A_{j}}=\sum_{i_{j}}p_{i_{j}}|\a_{i_{j}}\rangle\langle\a_{i_{j}}|_{A_{j}}$ is also an ensemble decomposition of $\r_{A_{j}}$.
According to the definition of convex roof coherence measure, there is $\sum_{i_{j}}p_{i_{j}}C_{f}(|\a_{i_{j}}\rangle_{A_{j}}) \geq \sum_{i_{j}}q_{i_{j}}C_{f}(|\psi_{i_{j}}\rangle_{A_{j}}) =  C_{f}(\r_{A_{j}})$. Then, we have proven Eq. \eqref{eq:An} for pure states,
\begin{equation}
    \label{eq:pAn}
C_{f}(|\a\rangle_{A_{1}A_{2}...A_{n}})\geq \sum_{j=1}^{n}\sum_{i_{j}}p_{i_{j}}C_{f}(|\a_{i_{j}}\rangle_{A_{j}} )\geq \sum_{j=1}^{n}C_{f}(\r_{A_{j}}).
\end{equation}

Next, we prove that Eq. \eqref{eq:An} holds for all states, and we use $\r_{A_{1}A_{2}...A_{n}}=\sum_{k}q_{k}|\psi_{k}\rangle\langle\psi_{k}|_{A_{1}A_{2}...A_{n}}$ to represent one of the optimal decompositions that gives $C_{f}(\r_{A_{1}A_{2}...A_{n}})$. Using Eq. \eqref{eq:pAn}, we have
\begin{equation}
    \label{eq:aAn}
    \begin{split}
        C_{f}(\r_{A_{1}A_{2}...A_{n}})& =\sum_{k}q_{k}C_{f}(|\psi_{k}\rangle_{A_{1}A_{2}...A_{n}})\\
        &\geq \sum_{k}q_{k}(C_{f}(\r^{k})_{A_{1}}+C_{f}(\r^{k})_{A_{2}}+...+C_{f}(\r^{k})_{A_{n}})\\
        &=\sum_{j=1}^{n}\sum_{k}q_{k}(C_{f}(\r^{k})_{A_{j}}).
    \end{split}
\end{equation}Here $(\r^{k})_{A_{j}}=\tr_{A_1\cdots A{j-1}A_{j+1}\cdots A_n}|\psi_{k}\rangle\langle\psi_{k}|_{A_{1}A_{2}...A_{n}}$. According to the convexity of a coherence measure, there is $\sum_{k}q_{k}C_{f}(\r^{k})_{A_{j}}\geq C_{f}(\sum_{k}q_{k}(\r^{k})_{A_{j}}) $, which lead to
\begin{equation}
    \label{eq:final}
    \begin{split}
        C_{f}(\r_{A_{1}A_{2}...A_{n}})&\geq \sum_{j=1}^{n}\sum_{k}q_{k}(C_{f}(\r^{k})_{A_{j}})\\
        &\geq \sum_{j=1}^{n}C_{f}(\sum_{k}q_{k}(\r^{k})_{A_{j}})=\sum_{j=1}^{n}C_{f}(\r_{A_{j}}).
    \end{split}
\end{equation}
This completes the proof.
\end{proof}

\section{Conditions for the equalities to hold}
\label{sec:equ}

This section investigates the conditions under which equality holds in coherence measure and superadditive relationships. We first analyze the equality conditions of the basic conditions of the coherence measure. After that, we describe the additivity conditions of bipartite states and give concrete examples.
\subsection{Conditions for the Equalities in Coherence Measures}
We investigate the conditions under which each inequality in Definition \ref{def:c} becomes an equality:

(i) Positive definiteness $ \ C(\rho)\geq 0$: The equation holds if and only if the state $\r$ is an incoherent state;

(ii) Monotonicity under incoherent operations:$ \ C(\rho)\geq C(\Lambda(\rho))$ if $\Lambda$ is an incoherent operation. The equation holds if and only if the state $\r$ is an incoherent state;
 
(iii) Monotonicity under selective measurements on average:$\ C(\rho)\geq \sum_{n}p_{n}C(\rho_{n})$, $p_{n}=Tr(K_{n}\rho K_{n}^{\dagger})$, $\rho_{n}=K_{n}\rho K_{n}^{\dagger}/p_{n}$, and $\Lambda(\rho)=\sum_{n}K_{n}\rho K_{n}^{\dagger}$ is an incoherent operation. The equation holds if and only if all states $\r_n$ and $\r$ have the same coherence measure value;

(iv) Convexity: $\sum_{n}q_{n}C(\rho_{n})\geq C(\sum_{n}q_{n}\rho_{n})$ for any set of states$\{\rho_{n}\}$ and probability distribution $\{q_{n}\}$. The equation holds if and only if all states $\r_n$ and $\r$ have the same coherence measure value.

\subsection{Conditions for the Equality of Superadditivity of Bipartite States}

We investigate the condition under which the inequality of Eq. \eqref{eq:AB} is saturated. We recall that the convex roof coherence measure is defined as $C_{f}(\varphi):=f(|c_{0}|^{2},|c_{1}|^{2},...,|c_{d-1}|^{2})$. Let $f:\bbR^d\rightarrow \bbR$ be a function defined in vectors. For any vectors $\boldsymbol{x}=(x_1,...,x_m)\in\bbR^m$, and $\boldsymbol{y}=(y_1,...,y_n)\in\bbR^n$, we have
\begin{equation}
\label{eq:xyn}
f(\boldsymbol{x})+f(\boldsymbol{y})=f(\boldsymbol{x}\otimes\boldsymbol{y}).
\end{equation}
Here, $\boldsymbol{x}\otimes\boldsymbol{y}=(x_1y_1,...,x_1y_n ,...,x_my_n)\in\bbR^{mn}$ denotes the Kronecker product of $\boldsymbol{x}$ and $\boldsymbol{y}$. 
\begin{theorem}
    
\label{cor:f}
Let $f$ be a function satisfying multiplicative separability in the form of Eq. \eqref{eq:xyn}. Then for the separable pure state $|\psi\rangle_{AB} = |\a\rangle \otimes |\b\rangle $, the pure states $|\a\rangle = \sum_i a_i |i\rangle_A$ and $|\b\rangle = \sum_j b_j |j\rangle_B$, with $\sum_i |a_i|^2 = 1$ and $\sum_j |b_j|^2 = 1$. The convex roof coherence measure $C_f$ satisfies the additivity
\begin{equation}
    \label{eq:eq}
C_f(|\psi\rangle_{AB})=C_f(|\a\rangle)+C_f(|\b\rangle).
\end{equation}
\end{theorem}
\begin{proof}
 The coherence measure can be expressed as
$$C_f(|\a\rangle) = f( |a_0|^2, |a_1|^2, \ldots, |a_{d_A-1}|^2),$$
$$C_f(|\b\rangle) = f( |b_0|^2, |b_1|^2, \ldots, |b_{d_B-1}|^2 ).$$
The tensor product of $|\a\rangle$ and $|\b\rangle$ are expressed as
$$|\psi\rangle_{AB} = |\a\rangle \otimes |\b\rangle = ( \sum_i a_i |i\rangle_A ) \otimes( \sum_j b_j |j\rangle_B ) = \sum_{i,j} a_i b_j |i\rangle_A |j\rangle_B.$$
The coherence measure can be expressed as

\begin{equation}
    \begin{aligned}
    C_f(|\a\rangle)+C_f(|\b\rangle)
        &=f( |a_0|^2, |a_0|^2 ,\ldots,|a_{d_A-1}|^2)+f(|b_0|^2,|b_1|^2,\ldots,|b_{d_B-1}|^2 )\\
        &=f( |a_0|^2 |b_0|^2, |a_0|^2 |b_1|^2,\ldots,|a_i|^2 |b_j|^2,\ldots,|a_{d_A-1}|^2 |b_{d_B-1}|^2 )\\
        &= f( |a_0 b_0|^2, |a_0 b_1|^2,\ldots,|a_i b_j|^2,\ldots,|a_{d_A-1} b_{d_B-1}|^2 )\\   
        &= C_f(\a\otimes\b)=C_f(|\psi\rangle_{AB}).
    \end{aligned}
\end{equation}
That is, we prove the additivity of separable states in Eq. \eqref{eq:eq}.
\end{proof}

\begin{corollary}
    \label{cor:multif}
    Let $f$ remain the function introduced earlier that satisfies the conditions of Theorem \ref{cor:f}. Then, for a pure state composed as the tensor product of the pure state of the $n$ subsystem, $|\psi\rangle_{1,2,\dots,n}=|\psi_1\rangle\otimes|\psi_2\rangle\otimes\cdots\otimes|\psi_n\rangle$. The convex roof coherence measure $C_f$ satisfies the additive
 \begin{equation}
        \label{eq:multieq}
C_f(|\psi\rangle_{1,2,\dots,n})=\sum_{i=1}^nC_f(|\psi_i\rangle)
    \end{equation}
    \end{corollary}

\begin{proof}
    We proceed with the proof by mathematical induction. First, for the case $n=2$, it reduces to the original Theorem \ref{cor:f}: given $|\psi\rangle_{12}=|\psi_1\rangle\otimes|\psi_2\rangle$, since the measure satisfies additivity, we clearly have for a pure product state $C_f(|\psi\rangle_{12})=C_f(|\psi_1\rangle\otimes|\psi_2\rangle)=C_f(|\psi_1\rangle)+C_f(|\psi_2\rangle)$.  Now assume as the induction hypothesis that for any pure product state of the $k<n$ subsystems, one already has
    $$C_f(|\psi_1\rangle\otimes|\psi_2\rangle\otimes\cdots\otimes|\psi_k\rangle)=\sum_{i=1}^kC_f(|\psi_i\rangle$$
Consider the $n$-subsystem case $|\psi\rangle_{1,\dots,n}=|\phi\rangle_{1,\dots,n-1}\otimes|\psi_n\rangle$, where $|\phi\rangle_{1,\dots,n-1}=|\psi_1\rangle\otimes\cdots\otimes|\psi_{n-1}\rangle$ is the pure product state formed by the first $n-1$ subsystems. By the additivity of $C_f$ on pure product states, we have for:
$$C_f(|\psi\rangle_{1,\dots,n})=C_f(|\phi\rangle_{1,\dots,n-1})+C_f(|\psi_n\rangle)$$
By the induction hypothesis, 
$$C(|\phi\rangle_{1,\dots,n-1})=\sum_{i=1}^{n-1}C(|\psi_i\rangle)$$
Substituting it into the above, we obtain Eq. \eqref{eq:multieq}. This completes the proof.

\end{proof}

\subsection{Examples of Equality to Hold}

Example 1: The coherence measure of formation is defined as
\begin{equation}
    C_{for}(\r)=\inf_{\{p_{i}|\varphi_{i}\rangle\}}\sum_{i}p_{i}C_{r}(|\varphi_{i}\rangle),
\end{equation}
where $C_{r}(|\varphi\rangle)=S(\Delta(|\varphi_{i}\rangle\langle\varphi_{i}|))$ with $S(\r)=-\trace{\r\log_{2}\r}$ being the von Neumann entropy. We use $\Delta (\r)=\sum_{i}\r_{ii}|i\rangle\langle i|$ to denote the diagonal part of $\r$. To discuss the conditions under which the equation holds, we consider the pure state $|\a\rangle = \sum_i a_i |i\rangle_A$ and $|\b\rangle = \sum_j b_j |j\rangle_B$, with $\sum_i |a_i|^2 = 1$ and $\sum_j |b_j|^2 = 1$. We have 
$$C_{for}(|\a\rangle)=-\sum_{i}|a_{i}|^{2}\log_{2}|a_{i}|^{2},$$
$$C_{for}(|\b\rangle)=-\sum_{j}|b_{j}|^{2}\log_{2}|b_{j}|^{2}.$$
For the separable state $|\psi\rangle_{AB}=|\a\rangle\otimes|\b\rangle$, the coherence measure can be expressed as
$$\begin{aligned}
    C_{for}(|\psi\rangle_{AB})&=-\sum_{i,j}|c_{ij}|^{2}\log_{2}|c_{ij}|^{2}\\
&=-\sum_{i}|a_{i}|^{2}\log_{2}|a_{i}|^{2}-\sum_{j}|b_{j}|^{2}\log_{2}|b_{j}|^{2}\\
& =C_{for}(|\a\rangle) + C_{for}(|\b\rangle).
\end{aligned}$$
It shows that for the coherence measure of formation, the separable state satisfies the Eq. \eqref{eq:eq}.

Example 2: The coherence measure of the convex roof based on the entropy $\frac{1}{2}$ is defined as
$C_{\frac{1}{2}}(\r)=\inf_{\{p_i||\varphi_i\rangle\}}\sum_i p_i C_{\frac{1}{2}}(|\varphi_i\rangle),$
where $C_{\frac{1}{2}}(|\varphi\rangle)=2\log_2(\sum_{i=1}^d|c_i|
) $, for $|\varphi\rangle=\sum_ic_i|i\rangle$. The function $f$ satisfies Eq. \eqref{eq:xyn}. We consider pure states $|\a\rangle = \sum_i a_i |i\rangle_A$ and $|\b\rangle = \sum_j b_j |j\rangle_B$, with $\sum_i |a_i|^2 = 1$ and $\sum_j |b_j|^2 = 1$. We have
$$C_{\frac{1}{2}}(|\a\rangle)=2\log_2(\sum_{i=0}^{d-1}|a_i|) ,$$
$$C_{\frac{1}{2}}(|\b\rangle)=2\log_2(\sum_{i=0}^{d-1}|b_i|). $$
For the separable state $|\psi\rangle_{AB}=|\a\rangle\otimes|\b\rangle$, the coherence measure can be expressed as
$$\begin{aligned}
    C_{\frac{1}{2}}(|\psi\rangle_{AB})&=2\log_2(\sum_{i,j}|c_{ij}|)\\
&=2\log_2(\sum_{i=0}^{d-1}|a_i|)+2\log_2(\sum_{i=0}^{d-1}|b_i|)\\
& =C_{\frac{1}{2}}(|\a\rangle) + C_{\frac{1}{2}}(|\b\rangle).
\end{aligned}$$
It shows that for the coherence measure of the convex roof based on the entropy $\frac{1}{2}$, the separable state satisfies the Eq. \eqref{eq:eq}.

\section{Conclusion}
\label{sec:con}

We proposed and proved the structural criterion for attaining convex roof coherence measure superadditivity in bipartite states. Then we extended this result to the pure states of the general multipartite states and to the mixed states under its convex roof extension. Furthermore, we give the sufficient conditions when the above superadditive inequalities are equal. We explain and verify it through typical examples, such as the coherence measure of formation. Finally, we summarized the full text and pointed out several research directions that need to be further explored. The study of the additivity of the mixed states is difficult. It is hard to apply this additivity rule directly because the mixed states make the superposition of the coherent characteristic of each part difficult. Future research can further explore the necessity and completeness of the proposed conditions and whether these conditions are sufficient but also necessary for certain specific coherent measures. Another direction is to extend the theory to other coherent measures and discuss their additivity.

\section*{Acknowledgements}
Authors were supported by the NNSF of China (Grant No.12471427) and the Fundamental Research Funds for the Central Universities (Grant No. 4303088).

\bibliographystyle{unsrt}

\bibliography{main}

\end{document}